\documentclass[letterpaper, 10pt, conference]{IEEEtran} 
\usepackage{blindtext, graphicx}
\renewcommand{\baselinestretch}{0.95}
\usepackage{amsmath} 
\usepackage{amssymb}
\usepackage[utf8]{inputenc}
\usepackage{graphicx}
\usepackage[hang,flushmargin]{footmisc}
\usepackage{lipsum}
\makeatletter
\newcommand{\algorithmfootnote}[2][\footnotesize]{%
  \let\old@algocf@finish\@algocf@finish
  \def\@algocf@finish{\old@algocf@finish
    \leavevmode\rlap{\begin{minipage}{\linewidth}
    #1#2
    \end{minipage}}%
  }%
}

\usepackage{cite}

\ifCLASSINFOpdf
\else
\fi

\usepackage{lipsum,multicol}
\usepackage{subfigure}
\usepackage{floatrow}
\usepackage{algorithm}
\usepackage{listings}
\usepackage{mdwmath}
\usepackage{amsfonts}
\usepackage{algpseudocode}
\usepackage{amsmath}
\usepackage{diagbox}
\usepackage{caption}
\usepackage{amsthm}

\DeclareMathOperator*{\argmax}{argmax}
\newcommand*{\argmaxl}{\argmax\limits}
\DeclareMathOperator{\proj}{proj}
\DeclareMathOperator{\rank}{rank}

\theoremstyle{plain}
\newtheorem{theorem}{Theorem}
\newtheorem{lemma}[theorem]{Lemma}
\newcommand{\CLASSINPUTtoptextmargin}{0.7in}
\newcommand{\CLASSINPUTbottomtextmargin}{1in}

\begin{document}
%
\title{QR Approximation for Massive MIMO Fronthaul Compression}
\vspace{0.05cm}
\author{\IEEEauthorblockA{Aswathylakshmi P and Radha Krishna Ganti\\}
    \IEEEauthorblockA{Department of Electrical Engineering\\
                      Indian Institute of Technology Madras\\
                      Chennai, India 600036\\
                      \{aswathylakshmi, rganti\}@ee.iitm.ac.in}
}



%


\maketitle

\begin{abstract}
\boldmath
Massive MIMO's immense potential to serve large number of users at fast data rates also comes with the caveat of requiring tremendous processing power. This favours a centralized radio access network (C-RAN) architecture that concentrates the processing power at a common baseband unit (BBU) connected to multiple remote radio heads (RRH) via fronthaul links. The high bandwidths of 5G make the fronthaul data rate a major bottleneck. Since the number of active users in a massive MIMO system is much smaller than the number of antennas, we propose a dimension reduction scheme based on low rank QR approximation for fronthaul data compression. Link level simulations show that the proposed method achieves more than 17$\times$ compression while also improving the error performance of the system through denoising. 
\end{abstract}


%
\IEEEpeerreviewmaketitle
\section{Introduction}
Massive MIMO base station, with its large number of antennas, has the ability to support many users simultaneously through spatial multiplexing. This improves spectral efficiency and increases the network capacity. However, the huge processing complexity that such a system entails makes the centralized radio access network (C-RAN) architecture a better choice. In C-RAN architecture, the base station is split into two parts: a pooled baseband unit (BBU) at a centralized location, and connected to several remote radio heads (RRH) distributed geographically, as shown in Fig.1. The pooling of baseband resources provides more processing power and with its potential for cooperative radio to reduce interference, C-RAN can allow a higher density of RRHs to be put in place at low additional costs to the network operators \cite{Checko2015}. Thus massive MIMO combined with the C-RAN architecture can potentially support the ultra-high data rates envisioned in 5G. But the tight latency constraints and high bandwidth of 5G impose a huge capacity demand on the fronthaul links between the BBU and RRH. For example, with 64 RRH antennas, the low-PHY functional split between the BBU and RRH requires a fronthaul data rate upto 236 Gbps for 100MHz bandwidth, with the latest standard of the Common Public Radio Interface, eCPRI\cite{ecpri2017}. Laying such high capacity optical fibres for each BBU-RRH link would drive up the cost too much for the network operators, therefore compression techniques become necessary. 

\begin{figure}[t!]
    \centering
    \includegraphics[scale=0.25]{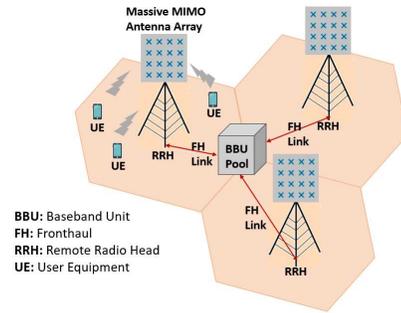}
    \caption{Massive MIMO C-RAN architecture: Centralized BBU connected to multiple RRHs with massive number of antennas}
    \label{fig:my_label}
\end{figure}

\par{In \cite{Peng2016}, four main approaches to uplink fronthaul compression are reviewed: point-to-point (P2P) compression, distributed source coding, compressed sensing (CS) and spatial filtering. While P2P compression and spatial filtering have low implementation complexity, if the signals received at different RRHs are highly correlated, distributed source coding performs better. CS based compression utilizes the sparsity of the uplink signals. In \cite{drvenica2016}, a lossy compression algorithm is  explored that applies FFT and Discrete  Cosine Transform (DCT) to the received signals and then discards low power frequency coefficients. The main drawback of this system is the need to divide the antenna array into many groups and apply the processing separately to each of them.}

\par{The Principal Component Analysis (PCA) compression algorithm proposed in \cite{Choi2016} uses the inherent sparsity of MIMO channels to reduce the number of links required in the fronthaul. It performs a low-rank approximation of the matrix consisting of the received signals by leveraging the signal correlation across space and time. But this requires computing the singular value decomposition (SVD) of the matrix, whose complexity is prohibitively high for large matrix dimensions as in the massive MIMO case, since the RRH has limited processing resources.}

\par{Another aspect that can drastically affect the data rate in the fronthaul link is the functional split between the BBU and RRH. The impact of the different functional splits (Fig.2) on the fronthaul rate and latency is evaluated in \cite{Chang2016}. The data rate is almost halved when moving from split A to B, as the cyclic prefix (CP) and guard bands are removed. The gains in moving to split C, where resource elements (RE) are demapped, is dependent on resource block utilization while for split D, it depends on the modulation order. For large modulation orders, split D can actually increase the data rate, as more bits are required to represent each sample. Finally, more than 90 percent reduction in the data rate can be achieved with split E compared to split A, but this comes at the cost of requiring all PHY layer processing to be at the RRHs increasing their complexity and decreasing flexibility. Though the data rate decreases from split A to E, the required control information increases.}

\begin{figure}[t!]
\vspace{3mm}
    \centering
    \includegraphics[scale=0.35]{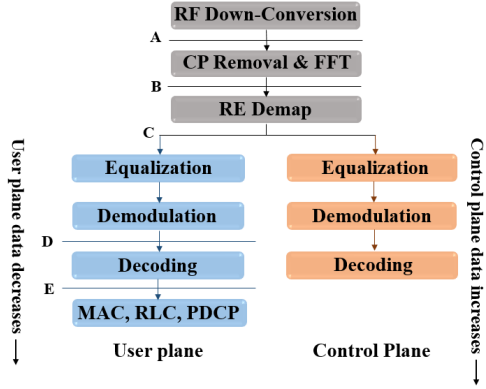}
    \caption{Possible functional splits between BBU and RRH in uplink}
    \label{fig:my_label}
\end{figure}

\par{In this paper, we propose to reduce the uplink fronthaul data rate in two stages. We choose split C in Fig.2, whereby CP and guard band removal and RE demapping is completed at the RRH. This reduces the data rate by almost half. To achieve further compression, we propose a low complexity algorithm based on QR decomposition for low rank approximation of the matrix composed of the complex baseband received signals. In addition to achieving high compression ratios, the proposed method provides denoising gain leading to better error performance compared to an uncompressed system.} 

\section{System Model}
We now provide a detailed description of the system model used in this paper. Consider a massive MIMO 5G base station with $N_r$ antennas at the RRH receiving signals from single antenna users in the uplink. We assume $N_u$ users in the system. In 5G, the uplink multiple access scheme is OFDMA \cite{38101}. Therefore, the bit-stream from each user undergoes M-QAM symbol mapping followed by OFDM modulation. OFDM modulation consists of sub-carrier mapping according to the resources allocated to the user, IFFT, and addition of a cyclic prefix (CP). These OFDM symbols pass through multi-path channel before reaching the RRH of the base station. The received signal at antenna $r$ at sampling instant $n$ is
\begin{equation}
y_{r}[n] = \sum_{u=1}^{N_u} x_{u}[n] \circledast h_{r,u}[n] + w_{r}[n] ,  
\end{equation}
where $x_u$ is the OFDM symbol from user $u$, $h_{r,u}$ is the multi-path channel response from user $u$ to antenna $r$,  $x_{u}[n] \circledast h_{r,u}[n]$ represents the convolution output between OFDM symbols of user $u$ and the multi-path channel response from user $u$ to antenna $r$, and $w_r$ is the additive white Gaussian noise (AWGN) with variance $\sigma$ at antenna $r$. 
\par{For our compression algorithm, we consider a block of $N$ time-domain samples received at the RRH antennas. We assume that the channel remains constant for the duration of these $N$ samples. The received signal matrix $\mathbf{Y}$ at the RRH is

\begin{equation*}
\mathbf{Y}=
  \begin{bmatrix}
   y_{1}[1] & y_{2}[1] & . & . & . & y_{N_r}[1] \\
    y_{1}[2] & y_{2}[2] & . & . & . & y_{N_r}[2]\\
    . & . & .  &   &   &. \\
    . & . &   &  . &   &. \\
    . & . &   &   & .  &. \\
    y_{1}[N] & y_{2}[N] & . & . & . & y_{N_r}[N]\\
  \end{bmatrix}_{N\times N_r} .
\end{equation*}Here, each column of $\mathbf{Y}$ represents the signal received at each antenna over a time span of $N$ samples.}

\section{Low Rank QR Approximation}

\begin{figure*}[h!]
    \centering
    \captionsetup{justification=centering}
    \includegraphics[scale=0.5]{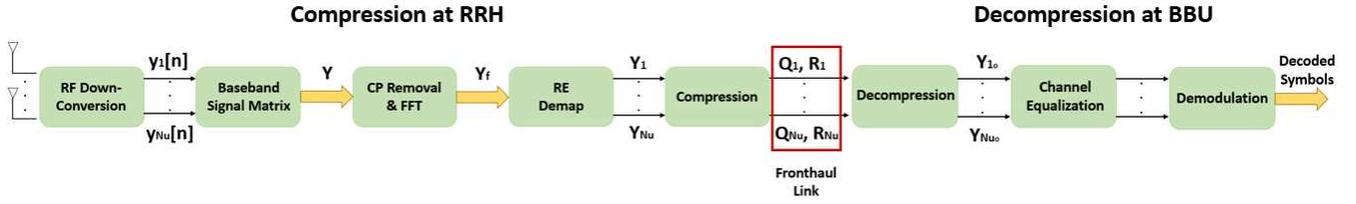}
    \caption{Proposed Compression Scheme}
    \label{fig:my_label}
\end{figure*}

We need to send the received signal matrix $\mathbf{Y}$ from the RRH to the BBU via the fronthaul link. Since the dimension of $\mathbf{Y}$, $N\times N_r$, is large in a massive MIMO setting, we aim to reduce its dimension to achieve compression. Assuming a maximum of $L$ multi-paths for each user, expanding (1), we have

\begin{equation}
\mathbf{Y}=
  \Big(\sum_{i=1}^{L}\mathbf{H_{i}}\mathbf{X_{i}} + \mathbf{W}\Big)^{T}.
\end{equation}

Here, 
\begin{equation*}
\mathbf{H_{i}}=
  \begin{bmatrix}
   h_{1,1}^{i} & h_{1,2}^{i} & . & . & . & h_{1,N_u}^{i} \\
    h_{2,1}^{i} & h_{2,2}^{i} & . & . & . & h_{2,N_u}^{i}\\
    . & . & .  &   &   &. \\
    . & . &   &  . &   &. \\
    . & . &   &   & .  &. \\
    h_{N_r,1}^{i} & h_{N_r,2}^{i} & . & . & . & h_{N_r,N_u}^{i}\\
  \end{bmatrix}_{N_{r}\times N_{u}}
\end{equation*}
is the matrix of complex channel gains for the $i^{th}$ multi-path,
\begin{small}
\begin{equation*}
\mathbf{X_{i}}=
  \setlength{\arraycolsep}{2pt}
  \renewcommand{\arraystretch}{0.8}
  \begin{bmatrix}
   x_{1}[1-i_1]&x_{1}[2-i_1]&.&.&.&x_{1}[N-i_1] \\
    x_{2}[1-i_2]&x_{1}[2-i_2]&.&.&.&x_{1}[N-i_2]\\
    .&.&. &  &  &. \\
    .&.&  & .&  &. \\
    .&.&  &  &. &. \\
    x_{N_u}[1-i_{N_u}] & x_{N_u}[2-i_{N_u}] & . & . & . & x_{N_u}[N-i_{N_u}]\\
  \end{bmatrix}_{N_{u}\times N} 
\end{equation*}
\end{small}is the matrix composed of transmitted symbols from the users, with user $u$'s symbols passing through the $i^{th}$ multi-path with delay $i_u$, and $\mathbf{W}$ is the $N_{r}\times N$ matrix of complex AWGN at the RRH. 
From (2), we observe that the received symbols at each antenna are correlated across time due to the correlation between the different $\mathbf{X_i}$'s. Jakes' one-ring model in \cite{jakes1994} shows that the channel gains in $\mathbf{H_i}$'s are also correlated as a function of the spacing and arrangement of the antennas at the receiver. Thus, the samples in $\mathbf{Y}$ have both spatial and temporal correlation, making $\mathbf{Y}$ a low-rank matrix. It is this nature of $\mathbf{Y}$ that we exploit to reduce its dimensions by applying a low-rank approximation. In particular, we choose QR decomposition to achieve this, which is a widely used algorithm that is both numerically stable and has a lower computational complexity than SVD \cite{sharma2013} because the RRH is resource constrained. We describe the procedure for obtaining low rank approximation of a matrix using QR decomposition below.
\par{Suppose $\mathbf{A_0}$ is the low rank approximation of a noisy matrix $\mathbf{A}$. We obtain $\mathbf{A_0}$ through QR decomposition using the Gram-Schmidt orthogonalization process. Then, $\mathbf{A_0} = \mathbf{QR}$, where $\mathbf{Q}$ is the matrix consisting of orthogonal basis vectors for the columns of $\mathbf{A}$, and $\mathbf{R}$ is an upper triangular matrix containing the projections of the columns of $\mathbf{A}$ onto these basis vectors. If we choose a subset containing $L$ columns of $\mathbf{A}$ to form a truncated basis, then $\mathbf{R}$ is upper triangular only upto column $L$. The choice of $L$ depends on the true rank of the noisy matrix $\mathbf{A}$, which can be determined using various methods \cite{kritchman2008}. Here we choose the $L$ columns according to their vector norms. In the case of the received signal matrix $\mathbf{Y}$, the norm of each column represents the total power received at the corresponding antenna. Therefore, we choose the antennas with the highest received powers to form the matrix of basis vectors, $\mathbf{Q}$. We assume $N_r < N$ and hence, the true rank of $\mathbf{Y}$ will be less than or equal to $N_r$, the $\min\{N, N_r$\}.}

\renewcommand\footnoterule{}      
\begin{algorithm}[b!] 
\caption{Low-rank approximation with QR decomposition}\footnotetext{$\mathbf{y_{u_n}}$ denotes column vector $n$ of $\mathbf{Y_u}$, $\mathbf{q_{i}}$ and $\mathbf{r_{i}}$ denote column vector $i$ of $\mathbf{Q_u}$ and $\mathbf{R_u}$, respectively. $\mathbf{e_{1}^{k}}$ denotes column vector of length $k$ with first component 1 and rest 0s.}
\label{alg:loop}
\begin{algorithmic}[1]
\For{$u \gets 1 \hspace{0.1cm} to \hspace{0.1cm} N_u$}
 \State {$\mathbf{q_{1}}\gets \argmaxl_{\mathbf{y_{u_n}}}\|\mathbf{y_{u_n}}\|_{2}, n\in\{1,2,...,N_r\}$}
 \State {$\mathbf{r_1} \gets \mathbf{e_{1}^{L_u}}$}
 \State {$\mathbf{Y_{u}^{1}} \gets \mathbf{Y_{u}} \backslash \mathbf{q_{1}}$}
 \For{$i \gets 2 \hspace{0.1cm} to \hspace{0.1cm} L_u$}    
  \State {$\mathbf{p_{i}}\gets \argmaxl_{\mathbf{y_{u_n}^{i-1}}}\|\mathbf{y_{u_n}^{i-1}}\|_{2}, n\in\{1,2,...,N_{r}-i+1\}$}
  \State {$\mathbf{q_i}=\mathbf{p_{i}}-\sum_{j=1}^{i-1}\proj_{\mathbf{p_j}}\mathbf{p_{i}}$}
  \State {$\mathbf{r_i}=\Big(\proj_{\mathbf{p_1}}\mathbf{p_{i}}...\proj_{\mathbf{p_{(i-1)}}}\mathbf{p_{i}} \hspace{0.2cm} \big(\mathbf{e_{1}^{(L_{u}-i+1)}}\big)^T\Big)^T$}
  \State {$\mathbf{Y_{u}^{i}} \gets \mathbf{Y_{u}^{i-1}} \backslash \mathbf{p_{i}}$}
 \EndFor
 \For {$i \gets L_{u}+1 \hspace{0.1cm} to \hspace{0.1cm} N_r$}
  \State {$\mathbf{p_{i}}\gets \mathbf{y_{u_{(i-L_u)}}^{L_u}}$}
  \State {$\mathbf{r_i}=\Big(\proj_{\mathbf{p_1}}\mathbf{p_{i}} \hspace{0.2cm} \proj_{\mathbf{p_{2}}}\mathbf{p_{i}}...\proj_{\mathbf{p_{L_u}}}\mathbf{p_{i}}\Big)^T$}
 \EndFor
 \State {$\mathbf{Q_u} \gets [ \mathbf{q_1}, \mathbf{q_2}, ..., \mathbf{q_{L_u}}], \mathbf{R_u} \gets [ \mathbf{r_1}, \mathbf{r_2}, ..., \mathbf{r_{N_r}}]$}
 \State \Return {$\mathbf{Q_u}, \mathbf{R_u}$}
\EndFor
\end{algorithmic}
\end{algorithm}

\section{Proposed Compression Method}
We apply the approximation described above to the received signal matrix $\mathbf{Y}$. We first remove the CP and guard-bands before applying the approximation, since these do not need to be sent to the BBU. For different users, the set of RRH antennas that offer the best signal-to-noise ratios (SNRs) will be different, as this depends on the location and orientation of the user with respect to the RRH. However, if we apply the algorithm to $\mathbf{Y}$ in the time-domain, where the users are not separated, the set of antennas chosen as the basis will be common to all users. Since we choose the antennas based on their total received powers, the users nearer to the RRH that contribute more power, will be favoured over the users farther from the RRH. In order to avoid this, we convert $\mathbf{Y}$ to the frequency domain by applying FFT so that we can separate the users according to the sub-carriers allocated to them. The sub-carrier allocation of each user is known to the base station. Then we choose the best antennas for each user to form the basis. 
\par{The process of compression at the RRH and decompression at the BBU is illustrated in Fig. 3. We explain each of the steps in detail below. After RF down-conversion, we construct the baseband signal matrix $\mathbf{Y}$ using the signals received at the $N_r$ antennas over a time span of $N$ symbols. Therefore, $\mathbf{Y}$ is of dimension $N\times N_r$. Without loss of generality, we choose $N$ to be the duration of one OFDM symbol. We first remove the CP and apply FFT to $\mathbf{Y}$ to convert it to the frequency domain signal matrix $\mathbf{Y_f}$. If the total number of sub-carriers allocated to all the users is $N_f$, then $\mathbf{Y_f}$ is of dimension $N_f \times N_r$. Next we perform resource element (RE) demapping, which separates the signals from the different users. For this, we divide $\mathbf{Y_f}$ into sub-matrices corresponding to different users according to their allocated sub-carriers. We denote the sub-matrix of user $u$ as $\mathbf{Y_u}$. If $N_{f_{u}}$ is the number of sub-carriers allotted to user $u$, then $\mathbf{Y_u}$ will be of dimension $N_{f_{u}}\times N_r$. Due to the antenna correlation described before, and assuming $N_{r}<N_{f_u}$, the true rank of $\mathbf{Y_u}$ will be less than $N_r$ and equal to the number of independent multi-paths in the channel for user $u$. We now apply the QR compression algorithm described in Algorithm 1 to each $\mathbf{Y_u}$. Thus for each user $u$, the low-rank approximated matrix $\mathbf{Y_{u_0}} = \mathbf{Q_{u}R_{u}}$. 
We choose $L_u$ antennas having the highest received powers from $\mathbf{Y_u}$ to form the columns of $\mathbf{Q_u}$, where $\rank(\mathbf{Y_u}) \leq L_u \leq N_r$.}

\begin{lemma}
For the system and compression scheme described above, the fronthaul compression ratio (CR) is given by
\begin{equation}
CR = \frac{N N_{r}b_{Q}}{\sum_{u=1}^{N_u}L_{u}(N_{f_u}+N_{r})b_{Q} + N_{u} N_{r}\log_{2}N_{r}}.
\end{equation}
\end{lemma}
\begin{proof}
In the absence of any compression, the samples of the received signal matrix $\mathbf{Y}$ (of dimension $N\times N_r$) quantized to $b_Q$ bits are sent to the BBU via the fronthaul link. Therefore, the number of bits before compression, $B_{org}$ is given by
\begin{equation}
B_{org} = NN_{r}b_{Q}.
\end{equation}
During compression, $\mathbf{Y}$ is converted to the frequency domain and divided into sub-matrices $\mathbf{Y_u}$ corresponding to each user $u$. $\mathbf{Y_u}$ is of dimension $N_{f_{u}}\times N_r$, where $N_{f_{u}}$ is the number of sub-carriers allotted to user $u$. Each $\mathbf{Y_u}$ is then approximated to the product of the matrices $\mathbf{Q_u}$ and $\mathbf{R_u}$ by QR approximation. $\mathbf{Q_u}$ is of dimension $N_{f_u} \times L_u$ and $\mathbf{R_u}$ of dimension $L_u \times N_r$, where $\rank(\mathbf{Y_u}) \leq L_u \leq N_r$. Uniform quantization of $b_{Q}$ bits is applied to the samples of each $\mathbf{Q_u}$ and $\mathbf{R_u}$. Therefore, the number of bits after compression, $B_{cmp}$ is given by
\begin{equation}
B_{cmp} = \sum_{u=1}^{N_u}(N_{f_u} L_u b_{Q} + L_u N_r b_{Q})
=\sum_{u=1}^{N_u}L_{u}(N_{f_u}+N_{r})b_{Q}.
\end{equation}
The order in which the columns of $\mathbf{Y_u}$ were chosen to construct $\mathbf{Q_u}$ and $\mathbf{R_u}$ also need to be sent for proper reconstruction of $\mathbf{Y_u}$ at the BBU. Since we need $\log_{2}N_r$ bits to represent the index of each of the $N_r$ antennas in $\mathbf{Y_u}$, this amounts to an overhead of $B_{ovh}$ bits given by
\begin{equation}
B_{ovh} = N_{u}N_{r}\log_{2}N_r.
\end{equation}
Thus, combining (4), (5) and (6), the fronthaul CR is
\begin{equation}
CR = \frac{B_{org}}{B_{cmp}+B_{ovh}}\hspace{0.1cm},
\end{equation}
which gives us (3).
\end{proof}

\par{At the BBU side, we reconstruct all the $\mathbf{Y_{u_0}}$ by taking the product of the corresponding $\mathbf{Q_{u}}$ and $\mathbf{R_{u}}$. For decoding, assuming the channel is known at the BBU, we perform zero-forcing equalization on each $\mathbf{Y_{u_0}}$. This is followed by joint decoding where we combine each user's symbols from all the $N_r$ antennas and demodulate the M-QAM symbols.}

\par{We can show that the computational complexity of the compression algorithm described in Algorithm 1 is $\mathcal{O}(N_{f}N_{r}L_{u})$. On the other hand, the SVD compression in \cite{Choi2016} applied to the time-domain signal matrix $\mathbf{Y}$ has a higher complexity of $\mathcal{O}(NN_{r}^{2})$, assuming $N_r < N$ \cite{kishore2017}. Even if the SVD algorithm were to be applied to each user sub-matrix $\mathbf{Y_u}$ after removing the CP and guard bands, its complexity is $\mathcal{O}(N_{f}N_{r}^{2})$, which is $N_{r}/L_{u}$ times higher than the complexity of our algorithm. When $L_{u} \ll N_{r}$, the complexity of our algorithm is significantly lesser than that of the SVD algorithm.} 

\begin{figure}[t!]
\vspace{3mm}
    \centering
    \includegraphics[scale=0.37]{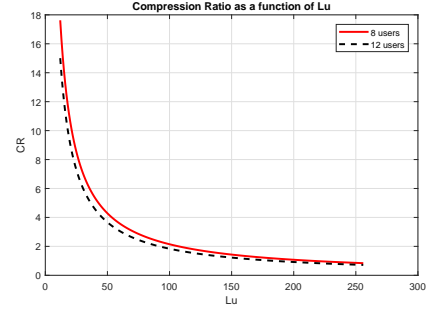}
    \caption{Compression Ratios (CRs) for the proposed method as a function of $L_u$, for 256 RRH antennas, 8 users and 12 users. The CR is inversely proportional to $L_u$, as observed from (3).}
    \label{fig:my_label}
\end{figure}

\section{Simulation Results}
We evaluate the Bit Error Rate (BER) performance of the proposed compression algorithm through Monte Carlo simulations using a massive MIMO uplink link level simulator in the baseband. We use the 3GPP tapped delay line (TDL) Rayleigh fading channel model for 5G, TDLA30\cite{38104}. We consider 100MHz bandwidth and 30kHz sub-carrier spacing, for which the FFT length is 4096 and CP length 288. Therefore, the length of one OFDM symbol is 4384, which is the number of time samples $N$ that we consider for one compression block. We use 64-QAM with 256 receive antennas at RRH. Thus, the dimension of the received signal matrix $\mathbf{Y}$ that is to be compressed is 4384$\times$256. The simulation parameters are summarized in Table I. 

\begin{table}[b!]
\centering
\begin{tabular}{ |p{3.5cm}||p{4cm}|  }
 \hline
 Modulation scheme & 64-QAM \\\hline
 No. of RRH antennas & 256\\\hline
 Bandwidth & 100 MHz\\\hline
 IFFT size & 4096\\\hline
 CP length & 288\\\hline
 Max. no. of RBs & 273\\\hline
 Size of one RB & 12\\\hline
 Channel model & TDLA30\\\hline
 Channel equalization & Zero forcing\\\hline
 Quantization & 15-bit uniform\\
  & (30 bits/complex sample)\\\hline
\end{tabular}
\caption{Simulation Parameters}
\end{table}

\begin{table}[b!]
\centering
\begin{tabular}{|l||*{3}{c|}}\hline
\backslashbox{$L_u$}{$N_u$}
&\makebox[3em]{8}&\makebox[3em]{12}\\\hline\hline
\makebox[3em]{12}&\makebox[3em]{17.4}&\makebox[3em]{14.5}\\\hline
\makebox[3em]{24}&\makebox[3em]{8.9}&\makebox[3em]{7.3}\\\hline
\end{tabular}
\caption{Achieved Compression Ratios for QR compression}
\end{table}

\begin{figure*}[ht!]
  \centering
  \mbox{
    \subfigure[\label{subfigure label}]{\centering\includegraphics[width=.4\linewidth]{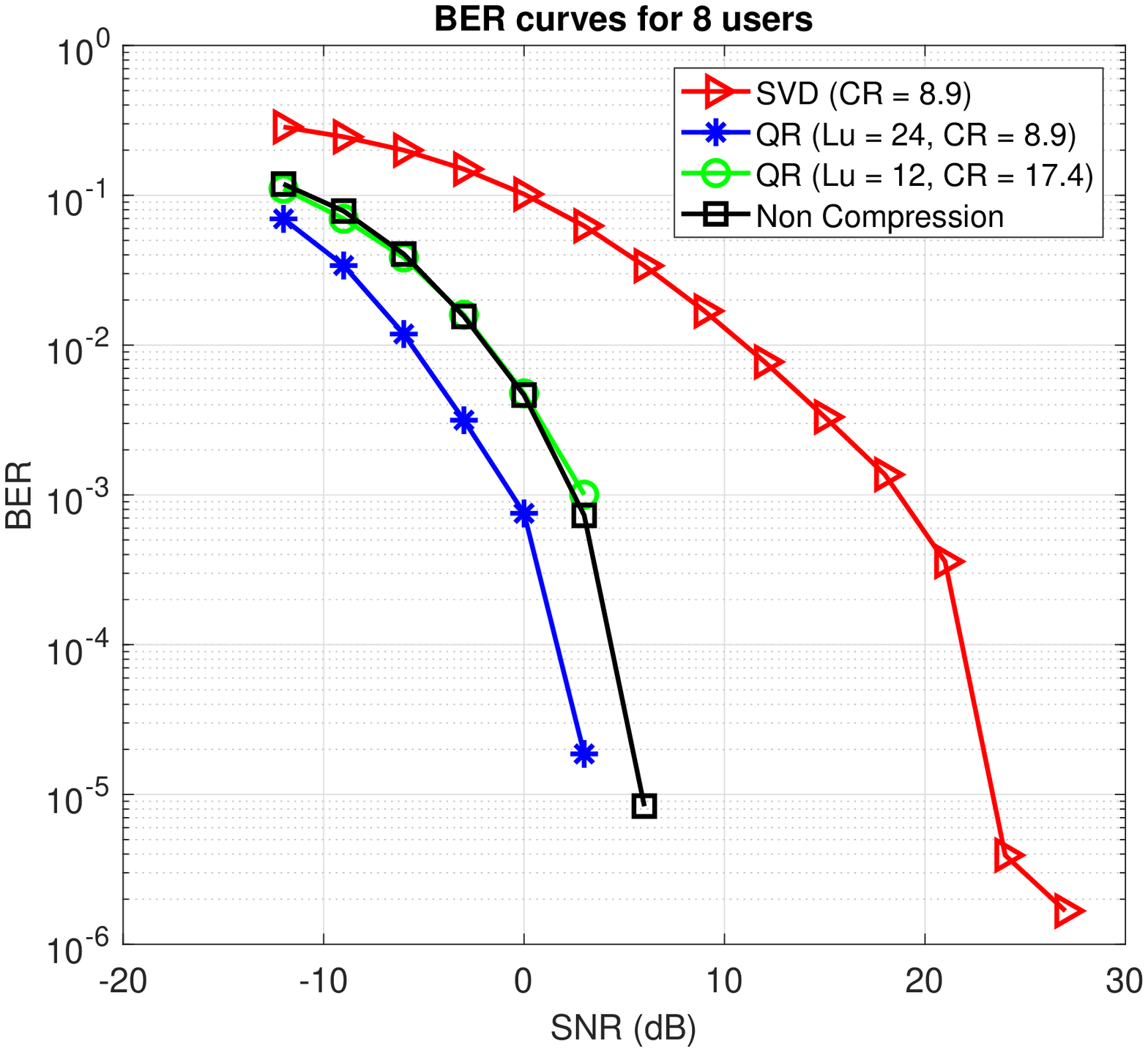}}\quad
    \subfigure[\label{subfigure label}]{\centering\includegraphics[width=.403\linewidth]{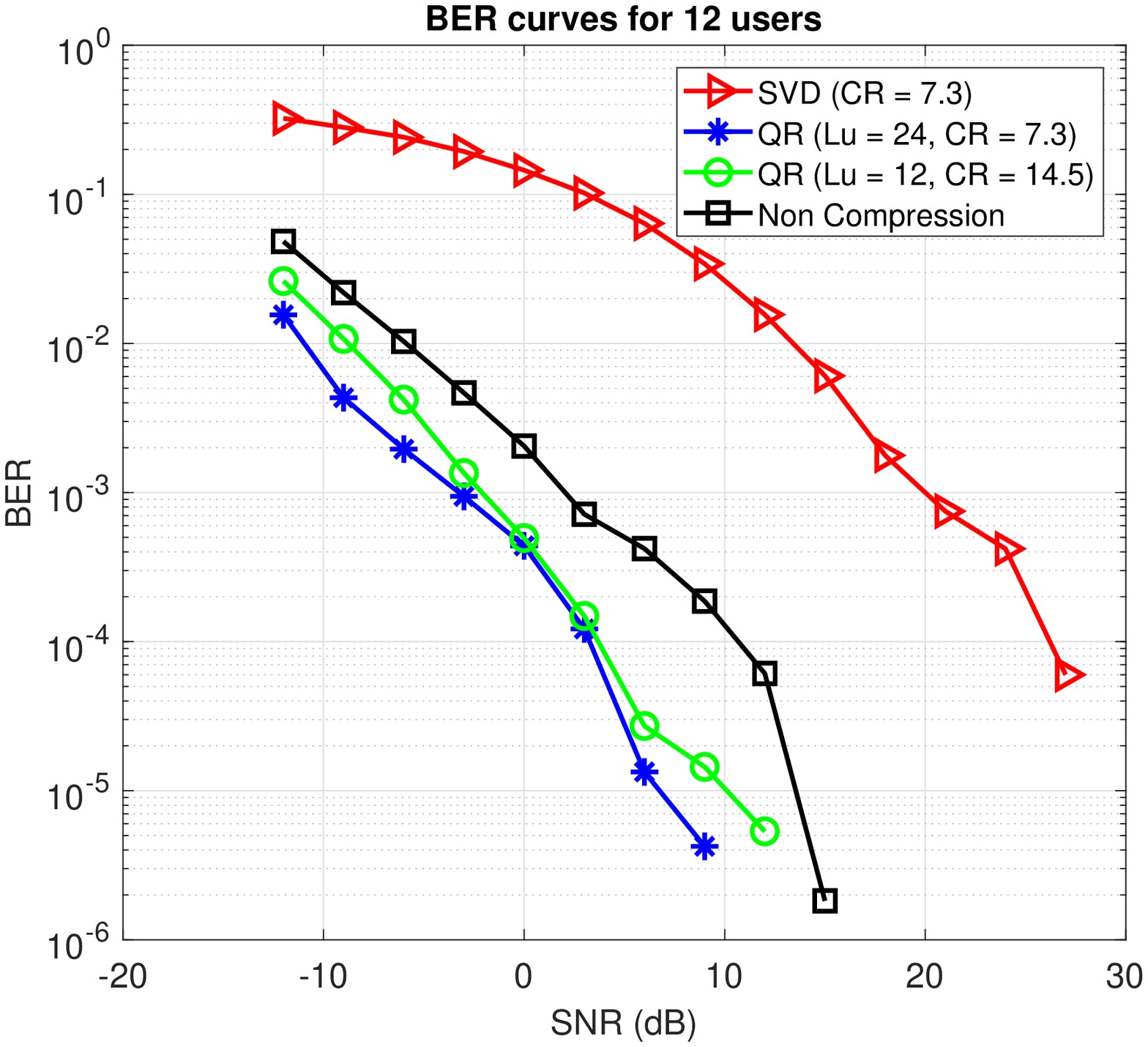}}\quad
    
  }
  \caption{Uncoded BERs of the proposed method with $L_u$= 12 and 24 compared with SVD compression and no compression, for (a) 8 users and (b) 12 users, for 256 RRH antennas with correlation coefficient 0.7 in the exponential correlation model.}
  \label{main figure label}
\end{figure*}

\par{We compare the BER performance of the proposed compression method against the SVD compression in \cite{Choi2016} and no compression, for an uncoded system. Assuming uniform linear array, we generate antenna correlation at the RRH according to the exponential correlation model \cite{loyka2001} with correlation coefficient 0.7, used for the compression proposed in \cite{Choi2016} for comparison. We consider two cases based on the number of users, (a) 8 users and (b) 12 users. The users are allocated resource blocks (RBs) according to their received SNRs at the RRH. The users with higher SNRs are allocated more RBs than those with lower SNRs. For 8 users, after arranging the users in the increasing order of their received SNRs at the RRH, the number of RBs allocated to them are 26, 28, 30, 32, 34, 36, 38, 40, respectively. Similarly, for 12 users, the RB allocation is 10, 12, 14, 16, 18, 20, 22, 24, 26, 28, 30, 32, respectively. In both the cases, the total number of RBs allocated should not exceed 273 as specified in \cite{38101}.}

\par{The true rank of each user sub-matrix in the frequency domain, $\mathbf{Y_u}$, is found to be 12, corresponding to the 12 taps in the multi-path channel model used. Fig. 4 shows the CRs achieved for different values of $L_u$. We see that lower the value of $L_u$, higher the CR achieved. In order to evaluate the impact of $L_u$ on the performance of the algorithm, we show the BER plots for 2 values of $L_u$, 12 and 24 in Fig. 5(a) (for 8 users) and Fig. 5(b) (for 12 users). We see that the algorithm performs better for the higher value of $L_u$. Thus, the choice of $L_u$ in our algorithm is a trade-off between the desired compression ratio (CR) and the required error performance. Table II shows the CRs achieved for $L_u$ = 12 and 24 for different number of users. We compare the performance of our algorithm against SVD compression in \cite{Choi2016} for compression ratio of 8.9 for 8 users (Fig.5(a)) and 7.3 for 12 users (Fig.5(b)). The SVD compression is applied to the time-domain received signal matrix $\mathbf{Y}$, as proposed in \cite{Choi2016}. The true rank of $\mathbf{Y}$ is 12 multiplied by the number of users, i.e, 96 for 8 users and 144 for 12 users. Therefore, in order to achieve the same CR as our algorithm, we need to reduce the total number of bits allocated for the samples in the SVD compression, which degrades its performance. Thus, we observe that our method performs better than the SVD method for both the 8 user and 12 user cases. We also plot the BER for the uncompressed system in both user cases. We see BER improvement for our method compared to the non-compression case due to the denoising gain of the low-rank approximation we apply \cite{zhang2009}.}

\section{Conclusion}
In this paper, we proposed a data compression scheme for massive MIMO fronthaul that combines an intra-PHY layer functional split between the BBU and RRH, and a dimension reduction algorithm at the RRH based on low rank QR approximation. Through link level simulations, we showed that the proposed method achieves 17.4$\times$ compression for the 8 user case and 14.5$\times$ compression for the 12 user case. In both the cases, the performance of the proposed scheme is better than no compression and the SVD compression in \cite{Choi2016}. The proposed method has a lower computational complexity than the SVD method and also has a denoising effect that improves its error performance.

\bibliographystyle{IEEEtrans}
\bibliography{ref}



\end{document}